\documentclass[a4paper,UKenglish,11pt]{article}

\usepackage[T1]{fontenc}
\usepackage[utf8]{inputenc}
\usepackage[margin=1.5in]{geometry}
\usepackage{graphicx}
\usepackage{amssymb,amsmath,amsthm}
\usepackage{epstopdf}
\usepackage{authblk}
\usepackage{wrapfig}
\usepackage{subfig}
\usepackage{hyperref}
\usepackage{appendix}
\usepackage{enumerate}
\usepackage{array}
\usepackage{delarray}
\usepackage{stmaryrd}
\usepackage{yfonts}
\usepackage{eufrak}
\usepackage{mathbbol}
\usepackage{tikz}
\usepackage{mathabx}
\usepackage{lscape}
\usepackage{changepage}
\usepackage{multicol}
\usetikzlibrary{fadings,decorations,decorations.pathreplacing,patterns,arrows}
\newcounter{theoremc}
\newcounter{tmp}
\newcounter{mytheorem}
\newcounter{mytheorem2}
\newcounter{lemmalarge}
\newcounter{lemmasmall}
\newcounter{lemmaroots}

\theoremstyle{plain}
\newtheorem{theorem}[theoremc]{Theorem}

\newtheorem{proposition}[theoremc]{Proposition}
\newtheorem{lemma}[theoremc]{Lemma}
\newtheorem{corollary}[theoremc]{Corollary}

\newtheorem*{statement}{Statement}

\theoremstyle{definition}
\newtheorem{definition}[theoremc]{Definition}

\theoremstyle{remark}
\newtheorem{remark}[theoremc]{Remark}

\newcommand{\N}{\mathbb{N}}
\newcommand{\Z}{\mathbb{Z}}

\newcommand{\1}{\mathbb{1}}
\newcommand{\0}{\mathbb{0}}
\renewcommand{\O}{\mathcal{O}}

\title{Kadanoff Sand Piles, following the snowball.\thanks{Partially supported by  IXXI (Complex System Institute, Lyon) and ANR projects Subtile and MODMAD. }}

\author[$\triangleright\triangleleft$]{K\'evin Perrot}
\author[$\diamond$]{\'Eric R\'emila}
\affil[$\triangleright$]{Universit\'e de  Lyon - LIP - (umr 5668 - CNRS - ENS de Lyon - Universit\'e Lyon 1) - 46 all\'ee d'Italie 69364 Lyon Cedex 7, France}
\affil[$\triangleleft$]{Universit\'e Nice–Sophia Antipolis - Laboratoire I3S - UMR 6070 CNRS - 2000 route des Lucioles, BP 121, F-06903 Sophia Antipolis Cedex, France}
\affil[$\diamond$]{Universit\'e de  Lyon - Groupe d'Analyse de la Th\'eorie Economique  Lyon Saint-Etienne - (umr  5824 CNRS - Universit\'e Lyon 2) - Site st\'ephanois, 6 rue Basse des Rives, 42 023 Saint-Etienne Cedex 2, France}
\affil[$~$]{\texttt{kevin.perrot@ens-lyon.fr}\hspace{1cm}\texttt{eric.remila@univ-st-etienne.fr}}
\date{}

\begin{document}

\maketitle

\begin{abstract}
This paper is about cubic sand grains moving around on nicely packed columns in one dimension (the physical sand pile is two dimensional, but the support of sand columns is one dimensional). The Kadanoff Sand Pile Model is a discrete dynamical system describing the evolution of a finite number of stacked grains ---as they would fall from an hourglass--- to a stable configuration. Grains move according to the repeated application of a simple local rule until reaching a stable configuration from which no rule can be applied, namely a fixed point.

The main interest of the model relies in the difficulty of understanding its behavior, despite the simplicity of the rule. We are interested in describing the shape of fixed point configurations according to the number of initially stacked sand grains. In this paper, we prove the emergence of a wavy shape on fixed points, {\em i.e.}, a regular pattern is (nearly) periodically repeated on fixed points. Interestingly, the regular pattern does not cover the entire fixed point, but eventually emerges from a seemingly highly disordered segment. Fortunately, the relative size of the part of fixed points non-covered by the pattern repetition is asymptotically null.\\

\begin{center}\textbf{Keywords}\end{center}\ \indent Sand pile model, discrete dynamical system, self-organized criticality, fixed point.
\end{abstract}


\section{Introduction}\label{s:introduction}

\subsection{Framework}

Understanding and proving regularity properties on discrete dynamical systems (DDS) is readily challenging, and inferring the global behavior of a DDS defined with local rules is at the heart of our comprehension of natural phenomena \cite{weaver,grauwin}. There exists a lot of simply stated conjectures, often issued from simulations, which remain open (for example the famous Langton's Ant \cite{propp,gajardo}). Sand pile models are a class of DDS defined by local rules describing how grains move in discrete space and time. We start from a finite number of stacked grains ---in analogy with an hourglass---, and try to predict the asymptotic shape of stable configurations. 

Bak, Tang and Wiesenfeld introduced sand pile models as systems presenting self-organized criticality (SOC), a property of dynamical systems having critical points as attractors \cite{bak88}. Informally, they considered the repeated addition of sand grains on a discretized flat surface. Each addition possibly triggers an avalanche, consisting in grains falling from column to column according to simple local rules, and after a while a heap of sand has formed. SOC is related to the fact that a single grain addition on a stabilized sand pile has a hardly predictable consequence on the system, on which fractal structures may emerge \cite{creutz96}. This model can be naturally extended to any dimension.

A one-dimensional sand pile configuration can be represented as a sequence $(h_i)_{i \in \N}$ of non-negative integers, $h_i$ being the number of sand grains stacked on column $i$. The evolution starts from the initial configuration $h$ where $h_0=N$ and $h_i=0$ for $i > 0$, and in the classical sand pile model a grain falls from column $i$ to column $i+1$ if and only if the height difference $h_i - h_{i+1} > 1$. One-dimensional sand pile models were well studied in recent years \cite{goles93,durandlose98,phan04,formenti07,phan08,PSSPM,formenti11}.

In this paper, we study a generalization of classical sand pile models, the Kadanoff Sand Pile Model. A fixed parameter, denoted $p$, is defined as the number of grains falling at each rule application. The results we expound in this paper present an interesting feature: we asymptotically completely describe the form of stable configurations, though there is a part of asymptotically null relative size which remains mysterious. Furthermore, proven regularities are directly issued from this messy part. This point and its link to the half-discrete half-continuous nature of sandpile models will be discussed in the conclusion.

We formally define the model in subsection \ref{ss:definition} and present a broader class of DDS in which it belongs in subsection \ref{ss:cfg}. Subsection \ref{ss:hourglass} exposes an interesting and useful way of computing fixed points, and subsection \ref{ss:objective} introduces the results we expound in section \ref{s:snowball}. At the light of those developments, section \ref{s:conclusion} discusses that sand pile models exhibit a behavior at the edge between discrete and continuous phenomena.

In some previous works (\cite{LATA}\cite{MFCS}), we obtained a similar result for the smallest parameter $p= 2$ (the case $p=1$ is the well known Sand Pile Model), but, for the general case, we have to introduce a completely different approach. The mains ideas are the following: we first relate different representations of a sand pile configuration (subsection \ref{ss:dds}), which leads to the construction of a DDS on $Z^{p+1}$ such that the orbit of a well chosen point (according to the number of grains $N$) describes the fixed point configuration we want to characterize. This system is quasi-linear in the sense that the image of a point is obtained by a linear contracting transformation followed by a rounding (in order to remain in $\Z^{p+1}$). We want to prove that this system converges rapidly, but the rounding makes the analysis  of the system very difficult (except for $p = 2$). The key-point (subsection \ref{ss:version}) is the reduction  of this system to another quasi-linear system on $\Z^{p}$, for which we have a clear intuition (subsection \ref{ss:study}), and which allows to conclude (subsection \ref{ss:wave}).


\subsection{Definition of the Kadanoff Sand Pile Model}\label{ss:definition}

Kadanoff {\em et al} proposed a generalization of classical models in which a  fixed parameter $p$ denotes the number of grains falling at each step \cite{kadanoff89}. Starting from the initial configuration composed of $N$ stacked grains on column 0, we iterate the following rule: if the difference of height between column $i$ and $i+1$ is greater than $p$, then $p$ grains can fall from column $i$, and one grain reaches each of the $p$ columns $i+1,i+2,\dots,i+p$. The rule is applied once during each time step. In the following, {\em column} and {\em index} are synonyms; and for the sake of imagery, we always consider indices (column numbers) to be increasing on the right (see figure \ref{fig:rule}).


\begin{figure}
  \centering \includegraphics{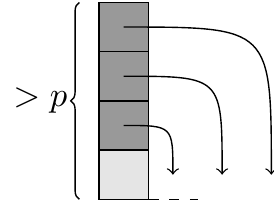}
  \caption{KSPM($p$) rule. When $p$ grains leave column $i$, the height difference $b_{i-1}$ is increased by $p$, $b_i$ is decreased by $p+1$ and $b_{i+p}$ is increased by 1. Other columns' height differences are not affected.}
  \label{fig:rule}
\end{figure}

Formally, it is defined on the space of ultimately null decreasing integer sequences where each integer represents a column of stacked sand grains. Let $h=(h_i)_{i \in \mathbb N}$ denote a {\em configuration} of the model, $h_i$ is the number of grains on column $i$.

In order to consider only the relative height between columns, we represent configurations as {\em sequences of height differences} $b=(b_i)_{i \in \N}$, where for all $i \geq 0,~ b_i=h_i-h_{i+1}$. This latter is the main representation we are using, within the space of ultimately null non-negative integer sequences.

We give two definitions of the model, obviously isomorphic. Definition \ref{definition1} is more natural, but Definition \ref{definition2} is more convenient and is the one we will use throughout the paper.

\newpage
\begin{multicols}{2}

\begin{definition}\label{definition1}
  The Kadanoff sand pile model with parameter $p>0$, KSPM($p$), is defined by two sets:
  \begin{itemize}
    \item \emph{configurations}. Ultimately null non-negative and decreasing integer sequences.
    \item \emph{transition rules}. We have a transition from a configuration $h$ to a configuration $h'$ on column $i$, and we note $h \overset{i}{\rightarrow} h'$ when
    \begin{itemize}
      \item $h'_i = h_i - p$
      \item $h'_{i+k} = h_{i+k} + 1$ for $0 < k \leq p$
      \item $h'_j = h_j$ for $j \not\in  \{i, i+1, \dots, i+p \}$. 
    \end{itemize}
  \end{itemize}
\end{definition}

We also say that $i$ is {\em fired}. Remark that according to the definition of the transition rules, $i$ may be fired if and only if $h_i-h_{i+1} > p$, otherwise $h'_i$ is negative or the sequence $h'$ is not decreasing.

\begin{definition}\label{definition2}
  The Kadanoff sand pile model with parameter $p>0$, KSPM($p$), is defined by two sets:
  \begin{itemize}
    \item \emph{configurations}. Ultimately null non-negative integer sequences.
    \item \emph{transition rules}. We have a transition from a configuration $b$ to a configuration $b'$ on column $i$, and we note   $b \overset{i}{\rightarrow} b'$ when
    \begin{itemize}
      \item $b'_{i-1}=b_{i-1} + p$ (for $i \neq 0$)
      \item $b'_i = b_i - (p-1)$
      \item $b'_{i+p} = b_{i+p} + 1$
      \item $b'_j = b_j$ for $j \not\in  \{i-1, i, i+p \}$. 
    \end{itemize}
  \end{itemize}
\end{definition}

Again, remark that according to the definition of the transition rules, $i$ may be fired if and only if $b_i > p$, otherwise $b'_i$ is negative.

\vspace{10pt}\ 

\end{multicols}

We  note $b \rightarrow b'$ when there exists an integer $i$ such that $b \overset{i}{\rightarrow} b'$. The transitive closure of $\rightarrow$ is denoted by  $\overset{*}{\rightarrow}$, and we say that $b'$ is {\em reachable} from $b$ when $b \overset{*}{\to} b'$. A basic property of the KSPM model is the \emph{diamond property} : if there exists $i$ and $j$ such that $b \overset{i}{\rightarrow} b'$ and $b \overset{j}{\rightarrow} b''$, then there exists a configuration $b'''$ such that $b' \overset{j}{\rightarrow} b'''$  and $b'' \overset{i}{\rightarrow} b'''$. 

We say that a configuration $b$ is \emph{stable}, or a \emph{fixed point}, if no transition is possible from $b$. As a  consequence of the diamond property, one can easily check that, for each configuration $b$, there exists a unique stable configuration, denoted by $\pi(b)$, such that  $b \overset{*}{\rightarrow} \pi(b)$. Moreover, for any configuration $b'$ such that $b \overset{*}{\rightarrow} b'$, we have $\pi(b') = \pi(b)$ (see \cite{goles02} for details). We denote $0^\omega$ the infinite sequence of zeros which is necessary to explicitly write the value of a configuration. An example of evolution is pictured on figure \ref{fig:example}. For convenience, we denote $N$ the initial configuration $(N,0^\omega)$, such that $\pi(N)$ is the height difference representation of the fixed point associated to the initial configuration composed of $N$ stacked grains (see figure \ref{fig:lattice} for an illustration). This paper is devoted to the study of $\pi(N)$ according to $N$.

\begin{figure}
  \begin{center}
    \includegraphics[width=\textwidth]{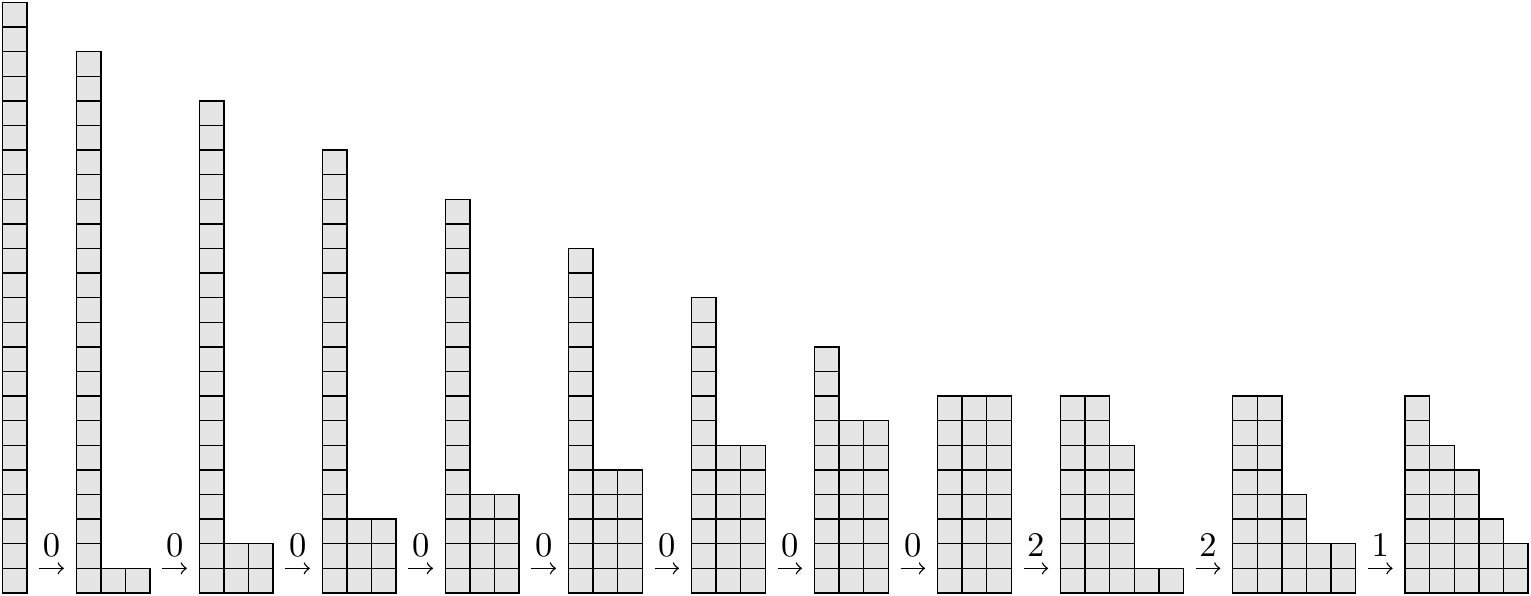}
  \end{center}
  \caption{A possible evolution in KSPM($2$) from the initial configuration for $N=24$ to $\pi(24)$. $\pi(24)=(2,1,2,1,2,0^\omega)$ and its shot vector is $(8,1,2,0^\omega)$.}
  \label{fig:example}
\end{figure}

\begin{figure}
  \begin{center}
    \includegraphics{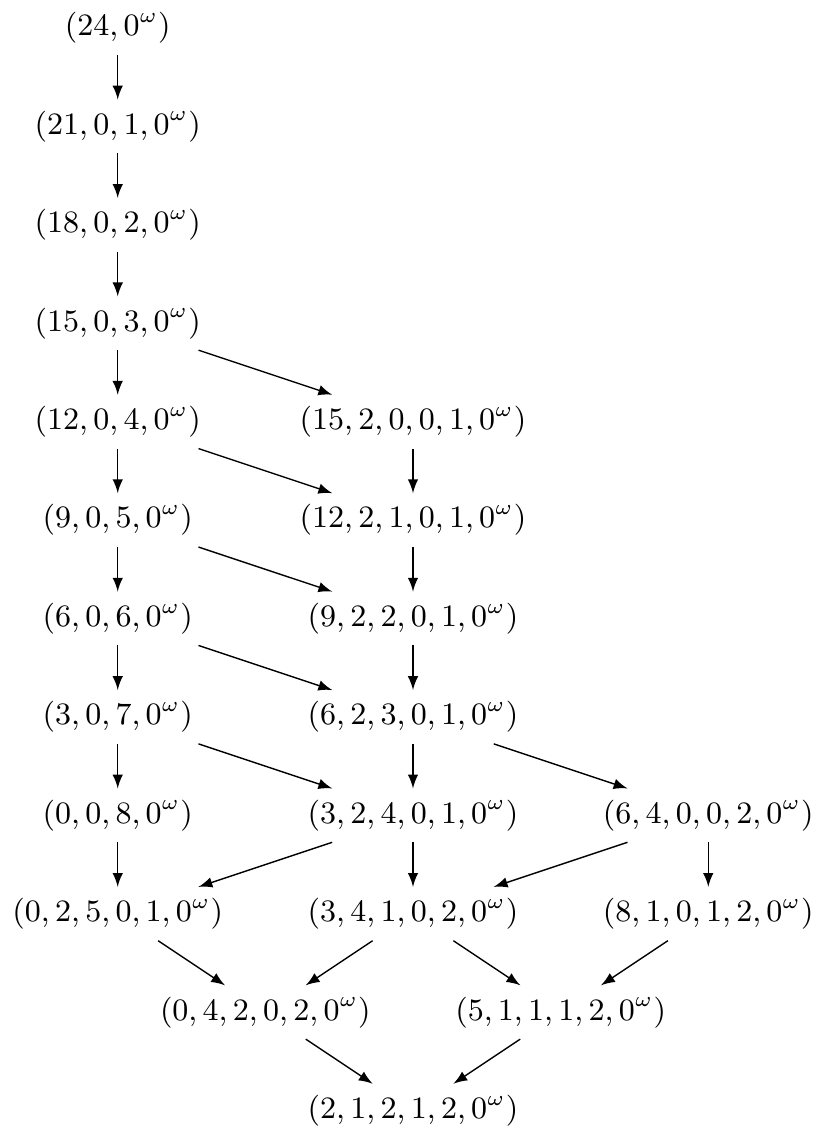}
  \end{center}
  \caption{The set of reachable configurations for $p=2$ and $N=24$. The initial configuration is on the top, and the unique fixed on the bottom.}
  \label{fig:lattice}
\end{figure}

A useful representation of a configuration reachable from $(N,0^\omega)$ is its {\em shot vector} $(a_i)_{i \in \N}$, where $a_i$ is the number of time that the rule has been applied on column $i$ from the initial configuration (see figure \ref{fig:example} for an example). This representation will play a major role in the following.


\subsection{Chip firing game}\label{ss:cfg}

Sand pile models are specializations of {\em Chip Firing Games} (CFG). A CFG is played on a directed graph in which each vertex $v$ has a load $l(v)$ and a threshold $t(v)=deg^+(v)\footnote{$deg^+(v)$ denotes the out-degree of $v$.}$, and the iteration rule is: if $l(v)\geq t(v)$ then $v$ gives one unit to each of its neighbors (we say $v$ is fired). As a consequence, we inherit all  properties of CFGs. 

Kadanoff sand pile is referred to as a {\em linear chip firing game} in \cite{goles02}. The authors show that the set of reachable configurations endowed with the order induced by the successor relation has a lattice structure, in particular it has a unique {\em fixed point}. Since the model is non-deterministic, they also prove \emph{strong convergence} {\em i.e.} the number of iterations to reach the fixed point is the same whatever the evolution strategy is. The morphism from KSPM(2) to CFG is depicted on figure \ref{fig:lcfg}.

\begin{figure}
  \begin{center}
    \includegraphics{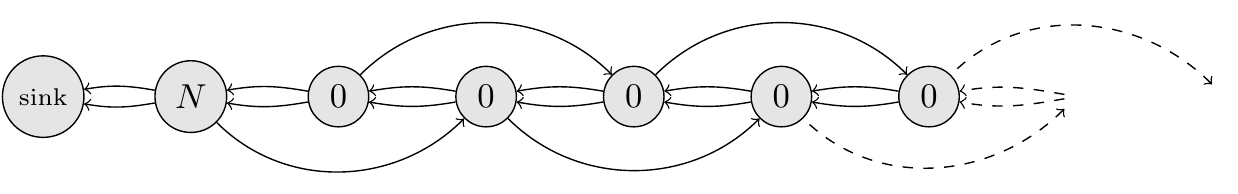}
  \end{center}
  \caption{The initial configuration of KSPM(2) is presented as a CFG where each vertex corresponds to a column (except the sink, vertices from left to right corresponds to columns $0,1,2,3,\dots$) with a load equal to the difference of height between column $i$ and $i+1$. For example the vertex with load $N$ is the difference of height between column 0 ($N$ grains) and column 1 ($0$ grain).}
  \label{fig:lcfg}
\end{figure}

When reasoning and writing formal developments about KSPM, it is convenient to think about its CFG representation where local rules let units of height differences move between columns.


\subsection{Hourglass, inductive computation and avalanches}\label{ss:hourglass}

In order to compute $\pi(N)$, the basic procedure is to start from the initial configuration $(N,0^\omega)$ and perform all the possible transitions. However, it also possible to start from the configuration $(0^\omega)$, add one grain on column 0 and perform all the possible transitions, leading to $\pi(1)$, then add another grain on column 0 and perform all the possible transitions, leading to $\pi(2)$, etc... And repeat this process until reaching $\pi(N)$.

Formally, let $b$ be a configuration, $b^{\downarrow 0}$ denotes the configuration obtained by adding one grain on column 0. In other words, if $b=(b_0,b_1,\dots)$ then $b^{\downarrow 0}=(b_0 +1 ,b_1,\dots)$. The correctness of the process described above relies on the fact that
$$(k,0^\omega) \overset{*}{\to} \pi(k-1)^{\downarrow 0}$$
Indeed, there exists a sequence of firings, named a {\em strategy}, $(s_i)_{i=1}^{i=l}$ such that $(k-1,0^\omega) \overset{s_1}{\to} \dots \overset{s_l}{\to} \pi(k-1)$. It is obvious that using the same strategy we have $(k,0^{\omega}) = (k-1,0^\omega)^{\downarrow 0} \overset{s_1}{\to} \dots \overset{s_l}{\to} \pi(k-1)^{\downarrow 0}$ since we only drag one more grain on column 0 along the evolution, which does not prevent any firing (see figure \ref{fig:inductive}). Thus, with the uniqueness of the fixed point reachable from $(k,0^\omega)$, we have the recurrence formula:
$$\pi(\pi(k-1)^{\downarrow 0})  =   \pi (k)$$
with the initial condition $\pi (0) = 0^\omega$, enabling an inductive computation of $\pi(k)$.

\begin{figure}
  \centering \includegraphics{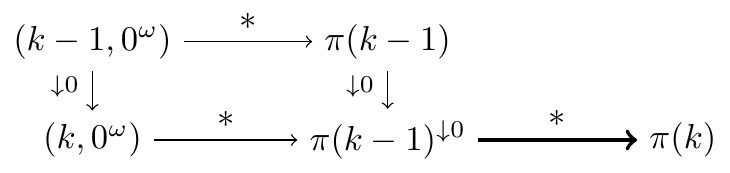}
  \caption{inductive computation of $\pi(k)$ from $\pi(k-1)$. The bold arrow represents an avalanche.}
  \label{fig:inductive}
\end{figure}

The strategy from $\pi(k-1)^{\downarrow 0}$ to $\pi(k)$ is called an {\em avalanche}. Note that from the non-determinacy of the model, this strategy is not unique. To overcome this issue, it is natural to distinguish a particular one which we think is the simplest: the {\em $k^{th}$ avalanche} $s^k$ is the leftmost strategy from $\pi(k-1)^{\downarrow 0}$ to $\pi(k)$, where {\em leftmost} is the minimal strategy according to the lexicographic order. This means that at each step, the leftmost possible firing is performed. A preliminary result of \cite{LATA} is that any column is fired at most once during an avalanche, which allows to write without ambiguity for an index $i$: $i \in s^k$ or $i \notin s^k$.

\begin{figure}
  \begin{center}
    \includegraphics{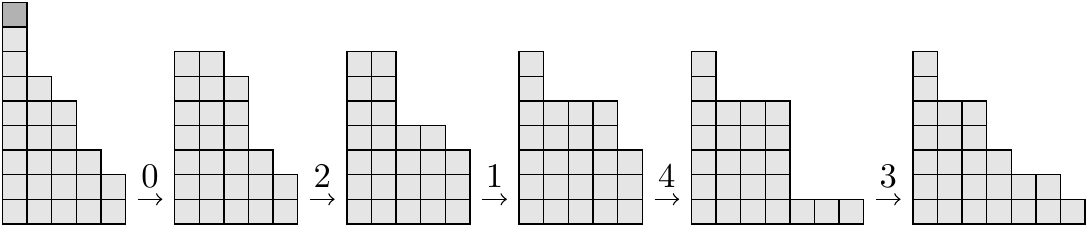}
  \end{center}
  \caption{An example of avalanche: starting from $\pi(24)$, we add one grain on column 0 (darkened on the leftmost configuration) and apply the iteration rule until reaching $\pi(25)$. Arrows are labelled by the index of the fired column (the leftmost unstable column is fired at each step).}
  \label{fig:avalanche25}
\end{figure}

A recent study (\cite{goles10}) showed that in the two dimensional case the avalanche problem (given a configuration and two indices $i$ and $j$, does adding one grain on column $i$ has an influence on index $j$?) on KSPM is \textbf{P}-complete, which points out an inherently sequential behavior.


\subsection{Objectives of the paper}\label{ss:objective}

Let us now introduce notations and state the results. Regarding regular expressions, let $\epsilon$ denotes the empty word, and $\cdot$ (respectively $+$) denote the {\em concatenation} (respectively {\em or}) operator. $\bigplus$ is to the {\em or} what $\Sigma$ is to the {\em sum}\footnote{$\bigplus \limits_{i=0}^{k} 0^i=\epsilon+0+00+\dots+\underset{k}{\underbrace{0\dots0}}$.}, and $^*$ is the Kleene star denoting finite repetitions of a regular expression (see for example \cite{hopcroft} for details). Finally, for a configuration $b$ we denote $b_{[n,\infty[}$ the infinite subsequence of $b$ starting from index $n$ to $\infty$. The study presented in this paper leads to the result:

\setcounter{mytheorem}{\value{theoremc}}
\begin{theorem}\label{theorem:main}
  There exists an $n$ in $\O(\log N)$ such that
  $$\pi(N)_{[n,\infty[} \in \left(\left(\bigplus \limits_{i=0}^{p+1} 0^i\right) \cdot p \cdot \ldots \cdot 2 \cdot 1\right)^*0^\omega$$
\end{theorem}
Some additional work will prove the more precise statement (named {\em Snowball Conjecture} in \cite{MFCS}):

\setcounter{mytheorem2}{\value{theoremc}}
\begin{theorem}\label{theorem:main2}
  There exists an $n$ in $\O(\log N)$ such that
  $$\pi(N)_{[n,\infty[} \in (p \cdot \ldots \cdot 2 \cdot 1)^*[0](p \cdot \ldots \cdot 2 \cdot 1)^*0^\omega$$
  where $[0]$ stands for at most one symbol 0.
\end{theorem}

Graphical representations of the Theorems are given on figures \ref{fig:wave1} and \ref{fig:wave2}. The example of $\pi(2000)$ for $p=4$ is given on figure \ref{fig:2000-h} of appendix \ref{a:2000}.

\captionsetup[subfigure]{width=.3\textwidth}
\begin{figure}
  \centering
  \subfloat[We call {\em wave} the pattern \hspace{.6cm}$p \cdot \ldots \cdot 2 \cdot 1$ in a sand pile configuration.]{\label{fig:wave}\hspace{1cm}\includegraphics{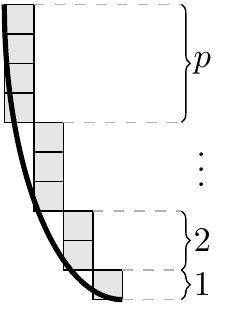}\hspace{1.5cm}}
  \subfloat[Theorem \ref{theorem:main}. Illustration of the regular expression $\left(\left(\bigplus \limits_{i=0}^{p+1} 0^i\right) \cdot p \cdot \ldots \cdot 2 \cdot 1\right)^*$ \hspace{.7cm} as a sand pile configuration.]{\label{fig:wave1}\includegraphics[width=4cm]{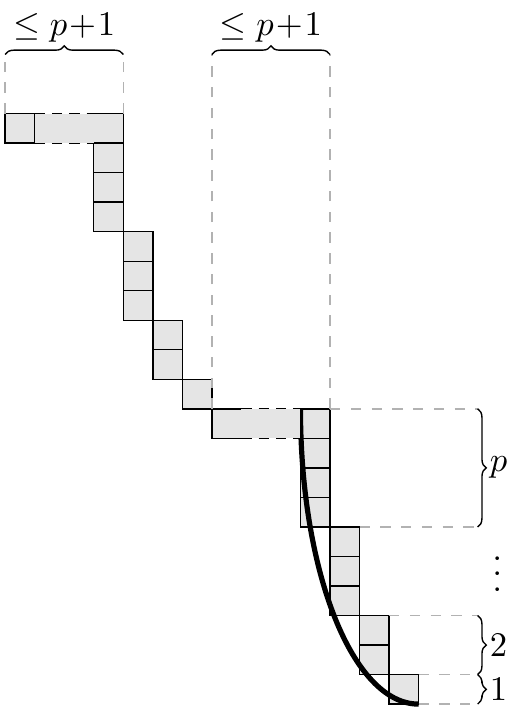}}
  \subfloat[Theorem \ref{theorem:main2}. From an index in $\O(\log N)$ the fixed point $\pi(N)$ consists of waves all consecutive to each other, except at at most one place where two waves may be separated by a 0.]{\label{fig:wave2}\hspace{.5cm}\includegraphics{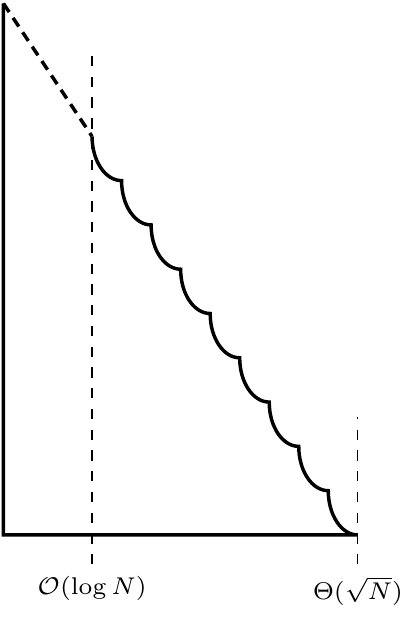}}
  \caption{A {\em wave} (figure \ref{fig:wave}), and graphical representations of Theorem \ref{theorem:main} (figure \ref{fig:wave1}) and Theorem \ref{theorem:main2} (figure \ref{fig:wave2}).}
\end{figure}


 
\section{Once upon a time, a snowball}\label{s:snowball}

We consider the parameter $p$ to be fixed. The proof of Theorem \ref{theorem:main} works as follows. We begin by establishing a relation between the height difference representation and the shot vector, leading to the construction of a DDS in $\Z^{p+1}$, such that the orbit of a well chosen point according to the number of grains $N$ describes the fixed point $\pi(N)$ (subsection \ref{ss:dds}). The aim is then to prove the convergence of this orbit in $\O(\log N)$ steps, such that the shot vector values it embeds involve the form described in Theorem \ref{theorem:main} (subsections \ref{ss:version}, \ref{ss:study} and \ref{ss:wave}).

\subsection{Linking height differences and shot vector in a dynamical system}\label{ss:dds}

In this subsection we construct a discrete dynamical system in $\Z^{p+1}$ such that the orbit of a particular point (chosen according to $N$) describes $\pi(N)$.

A fixed point $\pi(N)$ can be represented as a sequence of height differences $(b_i)_{i \in \N}$ ({\em i.e.}, $b_i=\pi(N)_i$ for all $i$) or as a shot vector $(a_i)_{i \in \N}$. Those two representations are obviously linked in various ways. In particular for any $n$ we can compute the height difference at index $n$ provided the number of firings at $n-p$, $n$ and $n+1$ because $b_n$ is initially equal to 0 (the case $n=0$ is discussed below), and: a firing at $n-p$ increases $b_n$ by 1; a firing at $n$ decreases $b_n$ by $p+1$; a firing at $n+1$ increases $b_n$ by $p$. Therefore, $b_n = a_{n-p}  - (p+1) a_n + p a_{n+1}$, with $ 0 \leq b_n  \leq p$ since $\pi(N)$ is a fixed point, and thus 
$$a_{n+1} =   -\frac{1}{p} a_{n-p} + \frac {p+1}{p} a_n +  \frac {1}{p}b_n$$

This equation expresses the value of the shot vector at position $n+1$ according to its values at positions $n-p$ and $n$, and a bounded perturbation $0 \leq \frac{b_n}{p} \leq 1$. We take as an initial condition $a_{-p}=N$ and $a_i=0$ for $-p<i<0$, representing the fact that the column 0 is the only one receiving $N$ units of height difference.

\begin{remark}\label{remark:determined}
Note that $a_{n+1} \in \N$, thus $-a_{n-p}+(p+1)a_n+b_n \equiv 0 \mod p$. As a consequence, the value of $b_n$ is {\em nearly determined}: given $a_{n-p}$ and $a_n$, there is only one possible value of $b_n$, except when $-a_{n-p} + (p+1) a_n \equiv 0 \mod p$ in which case $b_n$ equals $0$ or $p$.
\end{remark}

For example, consider $\pi(2000)$ for $p=4$ (see appendix \ref{a:2000}). We have $a_8=120$ and $a_4=189$, so $-a_4+5a_8=411 \equiv 3 \mod p$. From this knowledge, $b_8$ is determined to be equal to $1$, so that $a_9=-\frac{1}{4} a_4 + \frac{5}{4} a_8 +  \frac {1}{4}b_8=103$ is an integer.

We rewrite this relation as a linear system we can manipulate easily. $a_{n+1}$ is expressed in terms of $a_{n-p}$ and $a_n$, so we construct a sequence of vectors $(X_i)_{i \in \N}$ with $X_i \in \N^{p+1}$ and such that $X_n =(a_{n-p}, a_{n-p+1}, \dots, a_n)^T$ where $v^T$ stands for the transpose of $v$. Note that we consider only finite configurations, so there always exists an integer $n_0$ such that $X_n=\0$ for $n_0 \leq n$, with $\0=(0, \dots, 0)^T$.

Given $X_n$ and $b_n$ we can compute $X_{n+1}$ with the relation
$$X_{n+1}=A \, X_n + \frac {b_n}{p} J \hspace{.7cm}\text{with}\hspace{.7cm} A = \begin{pmatrix} 0 & 1 & & 0 & 0\\ & & \ddots & & \\ 0 & 0 & & 1 & 0\\ 0 & 0 & & 0 & 1\\ -\frac{1}{p} & 0 & & 0 & \frac {p+1}{p} \end{pmatrix} \hspace{.5cm} J=\begin{pmatrix}0\\\vdots\\0\\0\\1\end{pmatrix}$$
in the canonical base $B=(e_0,e_1,\dots,e_p)$, with $A$ a $p+1$ square matrix\footnote{As a convention, blank spaces are 0s and dotted spaces are filled by the sequence induced by its endpoints.}.

This system expresses the shot vector around position $n+1$ (via $X_{n+1}$) in terms of the shot vector around position $n$ (via $X_n$) and the height difference at $n$ (via $b_n$). Thus the orbit of the point $X_0=(N,0,\dots,0,a_0)$ in $\N^{p+1}$ describes the fixed point $\pi(N)$.

Note that it may look odd to study the sequence $(b_i)_{i \in \N}$ using a DDS that presupposes a knowledge on $(b_i)_{i \in \N}$. It is actually helpful because of the underlined fact that the value $b_n$ is {\em nearly determined} (remark \ref{remark:determined}): in the following we will make no supposition on the sequence $(b_i)_{i \in \N}$ (except that $0 \leq b_i \leq p$ for all $i$) and prove that the system converges exponentially quickly in $N$, such that from an $n$ in $\O(\log N)$ the sequence $(b_i)_{i \geq n}$ is {\em determined} to have a regular wavy shape.

The system we get is a linear map plus a perturbation induced by the discreteness of the values of height differences. Though the perturbation is bounded by a global constant at each step ($b_n \leq p$ for all $n$ since $\pi(N)$ is a fixed point), it seems that the non-linearity prevents classical methods to be conclusive concerning the convergence of this model.

We denote $\phi$ the corresponding transformation from $\Z^{p+1}$ to $\Z^{p+1}$, which is composed of two parts: a matrix and a perturbation. Let $R(x)=x^{p-1} + \frac{p-1}{p}x^{p-2} + \dots + \frac{2}{p} x+ \frac{1}{p}$, the characteristic polynomial of $A$ is $(1-x)^2R(x)$. We can first notice that $1$ is a double eigenvalue. A second remark, which helps to get a clear picture of the system, is that all the other eigenvalues are distinct and lesser than 1 from Lemma \ref{lemma:roots} (using a bound by Enerstr\"om and Kakeya \cite{enestromkakeya}). Therefore there exists a basis such that the matrix of $\phi$ is in Jordan normal form with a Jordan block of size 2. Then, we could project on the $p-1$ other components to get a diagonal matrix for the transformation, hopefully exhibiting an understandably contracting behavior. We won't exactly follow this approach, but will use these remarks in subsection \ref{ss:study}.

We tried to express the transformation $\phi$ in a basis such that its matrix is in Jordan normal form, but we did not manage to handle the effect of the perturbation expressed in such a basis. Therefore, we rather express $\phi$ in a basis such that the matrix and the perturbation act harmoniously. The proof of the main Theorem is done in three steps:
\begin{enumerate}
  \item the construction of a new dynamical system: we first express $\phi$ is a new basis $B'$, and then project along one component (subsection \ref{ss:version});
  \item the behavior of this new dynamical system is easily tractable, and we will see that it converges exponentially quickly (in $\O(\log N)$ steps) to a constant vector (subsection \ref{ss:study});
  \item finally, we prove that as soon as the vector is constant, then the wavy shape of Theorem \ref{theorem:main} takes place (subsection \ref{ss:wave}).
\end{enumerate}


\subsection{A new version of the dynamical system \dots}\label{ss:version}

From the dynamical system $X_{n+1}=A \, X_n +\frac {b_n}{p} J$ in the canonical basis $B$, we construct a new dynamical system for $\phi$ in two steps: first we change the basis of $\Z^{p+1}$ in which we express $\phi$, from the canonical one $B$, to a well chosen $B'$; then we project the transformation along the first component of $B'$. The resulting system on $\Z^p$, called {\em averaging system}, is very easily understandable, very intuitive, and the proof of its convergence to a constant vector can then be concluded straightforwardly. Let

$$
B'=
\begin{pmatrix}
1 & 0 & & 0\\
1 & 1 & & 0\\
\vdots & \vdots & \ddots & \\
1 & 1 & \dots & 1
\end{pmatrix}
\hspace{1cm}
B'^{-1}=
\begin{pmatrix}
  1 & & 0& 0\\
  -1 & \ddots & 0 & 0 \\
   & \ddots & \ddots & \\
   0 & & -1 & 1\\
\end{pmatrix}
$$
be square matrices\footnote{As a convention, blank spaces are 0s and dotted spaces are filled by the sequence induced by its extremities.} of size $p+1$. $B'=(e'_0,\dots,e'_p)$ (with $e'_i$ the $(i+1)^{th}$ column of the matrix $B'$) is a basis of $\Z^{p+1}$, and we have
$$
\begin{array}[t]{rrcl}
& B'^{-1} \, X_{n+1} & = & B'^{-1} \, A \, B' \, B'^{-1} \, X_n + \frac{b_n}{p} B'^{-1} \, J\\
\iff &X'_{n+1} & = & A' \, X'_n + \frac{b_n}{p} J'
\end{array}
$$
with
$$
X'_n = B'^{-1} \, X_n
\hspace{.7cm}
A' = B'^{-1} \, A \, B' =
\begin{pmatrix}
  1 & 1 & & 0\\
   & & \ddots & \\
   0 & 0 & & 1\\
   0 & \frac{1}{p} & \dots & \frac{1}{p}
\end{pmatrix}
\hspace{.5cm}
J' = B'^{-1} \, J =
\begin{pmatrix}
0\\ \vdots \\ 0\\ 1
\end{pmatrix}$$

We now proceed to the second step by projection along $e'_0$. Let $P$ denote the projection in $\Z^{p+1}$ along $e'_0$ onto $\{0\} \times \Z^p$. We can notice that $e'_0$ is an eigenvector of $A'$, hence projecting along $e'_0$ simply corresponds to erasing the first coordinate of $X'_i$. For convenience, we do not write the zero component of objects in $\{0\} \times \Z^p$.

The new DDS we now have to study, which we call {\em averaging system}, is
\begin{eqnarray}
  Y_{n+1} = M \, Y_n + \frac{b_n}{p} K \label{eq:averaging}
\end{eqnarray}
with the following elements in $\Z^p$ (in $\{0\} \times \Z^p$)
$$
Y_n = P \, X'_n
\hspace{1cm}
M=P \, A'=
\begin{pmatrix}
   0 & 1 & & 0\\
   & & \ddots & \\
   0 & 0 & & 1\\
   \frac{1}{p} & \frac{1}{p} & \dots & \frac{1}{p}
\end{pmatrix}
\hspace{1cm}
K=P \, J'=
\begin{pmatrix}
0\\ \vdots \\ 0\\ 1
\end{pmatrix}
$$


Let us look in more details at $Y_n$ and what it represents concerning the shot vector. We have $X_n=(a_{n-p},a_{n-p+1},\dots,a_n)^T$, thus
$$Y_n=P \, B'^{-1} \, X_n=
\begin{pmatrix}
a_{n-p+1}-a_{n-p}\\
\vdots\\
a_{n-1}-a_{n-2}\\
a_n-a_{n-1}
\end{pmatrix}
\text{ and for initialization }
Y_0=
\begin{pmatrix}
-N\\
0\\
\vdots\\
0\\
a_0
\end{pmatrix}$$
$Y_n$ represents differences of shot vector, which may of course be negative. In the next subsection \ref{ss:study} we will see that the averaging system is easily tractable and converges exponentially to a constant vector. Subsection \ref{ss:wave} then concentrates on implications following the presence of a constant vector {\em i.e.}, the creation and maintenance of a wavy shape.


\subsection{\dots which is easy to tackle \dots}\label{ss:study}

The averaging system is understandable in simple terms. From a vector $Y_n$ of $\Z^p$, $Y_{n+1}$ is obtained by
\begin{enumerate}
  \item shifting all the values one row upward;
  \item for the bottom component, computing the mean of values of $Y_n$, and adding a small perturbation (a multiple of $\frac{1}{p}$ between 0 and 1) to it.
\end{enumerate}

\begin{remark}
$(Y_i)_{i \in \N}$ are once more integer vectors, hence the perturbation added to the last component is again nearly determined: let $y_n$ denote the mean of the values of $Y_n$, we have $(y_n + \frac{b_n}{p}) \in \Z$ and $0 \leq \frac{b_n}{p} \leq 1$. Consequently, if $y_n$ is not an integer then $b_n$ is determined and equals $p(\lceil y_n \rceil - y_n)$, otherwise $b_n$ equals 0 or $p$.
\end{remark}

For example, consider $\pi(2000)$ for $p=4$ (see figure \ref{fig:2000-dsv} of appendix \ref{a:2000}, be careful that it pictures $a_n-a_{n+1}$ at position $n$). We have $Y_{13}=(-3,-5,-7,-7)^T$, then $y_{13}=-\frac{11}{2}$ and $b_{13}$ is forced to be equal to 2 so that $Y_{14}=(-5,-7,-7,-5)^T$ is an integer vector.


We can foresee what happens as we iterate this dynamical system and new values are computed: on a {\em large scale} ---when values are large compared to $p$--- the system evolves roughly toward the mean of values of the initial vector $Y_0$, and on a {\em small scale} ---when values are small compared to $p$--- the perturbation lets the vector wander a little around. Previous developments where intending to allow a simple argument to prove that this tiny wanderings does not prevent the exponential convergence towards a constant vector.

The study of the convergence of the averaging system works in three steps:
\begin{enumerate}
\item[(i).] state a linear convergence of the whole system; then express $Y_n$ in terms of $Y_0$ and $(b_i)_{0 \leq i \leq n}$;
\item[(ii).] isolate the perturbations induced by $(b_i)_{0 \leq i \leq n}$ and bound it by a constant;
\item[(iii).] prove that the other part (corresponding to the linear map $M$) is contracting exponentially quickly.
\end{enumerate}
From (ii) and (iii), a point evolves exponentially quickly to a ball of constant radius, then from (i) this point needs a constant number of more iterations in order to reach the center of the ball, that is, a constant vector.

For the sake of clarity, proofs of Lemmas are only sketched within the proof of the following Proposition, and are fully proven just below.

\begin{proposition}\label{lemma:average}
  There exists an $n$ in $\O(\log N)$ such that $Y_n$ is a constant vector.
\end{proposition}

\begin{proof}
Let $y_n$ (respectively $\overline{y}_n$, $\underline{y}_n$) denote the mean (respectively maximal, minimal) of values of $Y_n$. We will prove that $\overline{y}_n-\underline{y}_n$ converges exponentially quickly to 0, which proves the result.

We start with $Y_0=(-N,0,\dots,0,a_0)^T$, thus $\overline{y}_0-\underline{y}_0 = N+a_0 \leq \frac{p+1}{p} N$ since $a_0 \leq \frac{N}{p}$ (recall that $a_0$ is the number of times column 0 has been fired).

This proof is composed of two parts. Firstly, the system converges exponentially quickly on a large scale. Indeed, when $\overline{y}_n-\underline{y}_n$ is large compared to $p$, it is easy to be convinced that shifting the values and padding with the mean value plus a small perturbation decreases exponentially the gap between the values (because the perturbation is negligible). In other words, the value of $\overline{y}_n-\underline{y}_n$ decreases exponentially to a ball of constant radius:

\begin{adjustwidth}{.5cm}{0cm}
\setcounter{lemmalarge}{\value{theoremc}}
\begin{lemma}\label{lemma:large}
  There exists a constant $\alpha$ and a $n_0$ in $\O(\log N)$ s.t. $\overline{y}_{n_0} - \underline{y}_{n_0} < \alpha$.
\end{lemma}
\begin{proof}[Proof sketch]
Complete proof below. Since $Y_n$ converges towards the mean of its values, we consider the evolution of the {\em distance to the mean vector} associated to $Y_n$. We express $Y_n$ in terms of $Y_0$ and $(b_i)_{0 \leq i \leq n}$, and study separately the parts involving respectively $Y_0$ (part 1), and $(b_i)_{0 \leq i \leq n}$ (part 2). We first prove that part 1 converges exponentially quickly to $\0$ since it consists in multiplying $Y_0$ by the $n^{th}$ power of a matrix which eigenvalues are the set of roots of $R(x)$ (see Lemma \ref{lemma:roots}) and 0, thus a contracting matrix. Then we easily upper bound part 2 by a constant $\alpha$ independent of $N$, which completes the proof.
\end{proof}
\end{adjustwidth}

\begin{adjustwidth}{.5cm}{0cm}
\setcounter{lemmaroots}{\value{theoremc}}
\begin{lemma}\label{lemma:roots}
  Let $R(x)=x^{p-1} + \frac{p-1}{p}x^{p-2} + \dots + \frac{2}{p} x+ \frac{1}{p}$.\\ $R(x)$ has $p-1$ distinct roots $\lambda_1,\dots,\lambda_{p-1}$ and for all $i$, $\lambda_i \leq \frac{p-1}{p}$.
\end{lemma}
\begin{proof}[Proof sketch]
Complete proof below. It uses a bound by Enestr\"om and Kakeya (see for example\cite{enestromkakeya}), and the fact that $R(x)$ and $R'(x)$ are coprime.
\end{proof}
\end{adjustwidth}

Secondly, on a small scale, we use the fact that the system converges linearly (values of $Y_n$ are integers).

\begin{adjustwidth}{.5cm}{0cm}
\setcounter{lemmasmall}{\value{theoremc}}
\begin{lemma}\label{lemma:small}
  The value of $\overline{y}_n - \underline{y}_n$ decreases linearly: if $\underline{y}_n \neq \overline{y}_n$, then there is an integer $c$, with $0 \leq c \leq p$ such that $\overline{y}_{n+c} - \underline{y}_{n+c} < \overline{y}_n - \underline{y}_n$.
\end{lemma}
\begin{proof}[Proof sketch]
Complete proof below. The result follows the observation of the strict inequalities $\underline{y}_n < y_n < \overline{y}_n$. $y_n$ is used for the computation of the bottom component of $Y_{n+1}$, hence the idea is that all the values of $Y_{n+p}$ (recall the shifting of values) lies strictly in between the values of $Y_n$.
\end{proof}
\end{adjustwidth}

To conclude, we start with $\overline{y}_0-\underline{y}_0$ in $\O(N)$, we have a constant $\alpha$ and a $n_0$ in $\O(\log N)$ such that $\overline{y}_{n_0}-\underline{y}_{n_0} < \alpha$ thanks to the exponential decrease on a large scale (Lemma \ref{lemma:large}). Then after $p$ iterations the value of $\overline{y}_{n+p}-\underline{y}_{n+p}$ is decreased by at least 1 (Lemma \ref{lemma:small}), hence there exists $\beta$ with $\beta \leq p\alpha$ such that after $\beta$ more iterations we have $\overline{y}_{n_0+\beta}-\underline{y}_{n_0+\beta}=0$. Thus $Y_{n_0+\beta}$ is a constant vector, and $n_0+\beta$ is in $\O(\log N)$.
\end{proof}

In this proof, neither the discrete nor the continuous studies are conclusive by themselves. One one hand, the discrete study gives a linear convergence but not an exponential convergence. On the other hand, the continuous study gives an exponentially convergence towards a constant vector, but in itself the continuous part never reaches the constant vector but tends asymptotically towards it. It is the simultaneous study of those modalities (discrete and continuous) that allows to conclude.

\begin{remark}\label{remark:p2}
  Note that for $p=1$, the averaging system has a trivial dynamic. For $p=2$, the behavior is a bit more complex, but major simplifications are found: the computed value is equal to the mean of 2 values, hence the difference $\underline{y}_n-\overline{y}_n$ decreases by a factor of 2 at each time step. The key simplification arising for $p=2$ is the fact that this decrease arises straightforwardly {\em at each time step}.
\end{remark}

\setcounter{tmp}{\value{theoremc}}
\setcounter{theoremc}{\value{lemmalarge}}
\begin{lemma}
\setcounter{theoremc}{\value{tmp}}
  There exists a constant $\alpha$ and a $n_0$ in $\O(\log N)$ s.t. $\overline{y}_{n_0} - \underline{y}_{n_0} < \alpha$.
\end{lemma}
\begin{proof}
  We start with $Y_0=(-N,0,\dots,0,a_0)^T$, thus $\overline{y}_0-\underline{y}_0=N+a_0 \leq \frac{p+1}{p}N$.
  
  The relation linking $Y_n$ to $Y_{n+1}$ is
  $$Y_{n+1} = M \, Y_n + \frac{b_n}{p} \, K$$
  
  Since we want to prove that $Y_n$ converges to a constant vector close to the mean of its values, we want to consider the evolution of the distance to the mean vector associated to $Y_n$ via the following application $D$:
  $$Y_n=\begin{pmatrix}y_{n_0}\\\vdots\\y_{n_{p-1}}\end{pmatrix} \overset{D}{\longrightarrow} \begin{pmatrix}y_{n_0}-y_n\\\vdots\\y_{n_{p-1}}-y_n\end{pmatrix} \,\,\,\,\,\,\,\,\,\text{ with }\,\,\,\,\,\,\,\,\, y_n=\frac{1}{p}\sum \limits_i y_{n_i}$$
  
  The aim is thus to prove that there exists an $n_0$ in $\O(\log N)$ such that the norm of $D \, Y_{n_0}$ is bounded by a constant.
  
  We can notice that $Y_n-D \, Y_n=(y_n,\dots,y_n)^T$, therefore $M \, (Y_n - D \, Y_n) = (Y_n - D \, Y_n)$, and we have $(D \, M \, Y_n) - (D \, M \, D \, Y_n) = \0$. Consequently, we can consider the relation
  $$D \, Y_{n+1} = D\, M \, D \, Y_n + \frac{b_n}{p} \, D\, K$$
and express $Y_n$ in terms of $Y_0$ and $(b_i)_{0 \leq i \leq n}$:
  $$D \, Y_{n} = (D\,M)^n D \, Y_0 + \frac{1}{p} \sum \limits_{i=0}^{n-1} b_i (D \, M)^{n-1-i} D \, K$$
   
  In order to prove the result, we will see that the linear map $Z \mapsto (D \, M) \, Z$ is {\em eventually contracting}, hence it converges exponentially quickly to $\0$, its unique fixed point (\cite{hasselblatt} Corollary 2.6.13). That is, $(D \, M)^n \, D \, Y_0$ converges to $\0$ exponentially quickly. It then remains to upper bound the norm of the remaining sum by $\alpha$ to get the result.
  
  To prove that the map $Z \mapsto (D \, M) \, Z$ is eventually contracting, it is enough to prove that its {\em spectral radius}\footnote{the maximal absolute value of an eigenvalue of $D \, M$.} is smaller than 1 (\cite{hasselblatt} Corollary 3.3.5). This part is detailed in Lemma \ref{lemma:eigenDM} below, using the fact that $M$ is a companion matrix which eigenvalues are simply upper bounded with a result by Enestr\"om and Kakeya \cite{enestromkakeya}.

  Since $\overline{y}_0-\underline{y}_0$ is in $\O(N)$, $\| D \, Y_0 \|_{\infty}$ is also in $\O(N)$ and there exists an $n_0$ in $\O(\log N)$ such that $\| (D \, M)^{n_0} \, D \, Y_0 \|_{\infty} < 1$.
  
  It remains to upper bound the summation by a constant (we recall that for a matrix $A$, $\| A \|_\infty = \sup \| A \, x\|_\infty$ for $\|x\|_\infty=1$):
  $$\begin{array}{rcl}
  \left\| \frac{1}{p} \sum \limits_{i=0}^{n_0-1} b_i (D \, M)^{n_0-1-i} D \, K \right\|_{\infty} &\leq& \frac{1}{p}\sum \limits_{i=0}^{n_0-1} p\| D \, M \|_\infty^{n_0-1-i} \| D \, K\|_\infty\\
  &\leq& \frac{1}{1-\|D\,M\|_\infty} \| D\,K\|_\infty\\
  &\leq& \beta-1
  \end{array}$$
  for some constant $\beta$ independent of $N$. Finally, we have
  $$\| D \, Y_{n_0}\|_\infty \leq \| (D \, M)^{n_0} D \, Y_0 \|_\infty + \| \frac{1}{p} \sum \limits_{i=0}^{n-1} b_i (D \, M)^{n-1-i} D \, K \|_\infty \leq \beta$$
  and the fact that $\overline{y}_{n_0} - \underline{y}_{n_0} \leq 2 \| D \, Y_{n_0} \|_\infty$ completes the proof with $\alpha=2\beta$.
\end{proof}

\setcounter{tmp}{\value{theoremc}}
\setcounter{theoremc}{\value{lemmaroots}}
\begin{lemma}
\setcounter{theoremc}{\value{tmp}}
  Let $R(x)=x^{p-1} + \frac{p-1}{p}x^{p-2} + \dots + \frac{2}{p} x+ \frac{1}{p}$.\\ $R(x)$ has $p-1$ distinct roots $\lambda_1,\dots,\lambda_{p-1}$ and for all $i$, $\lambda_i \leq \frac{p-1}{p}$.
\end{lemma}

\begin{proof}
  The distinctness of the roots of $S(x)=px^{p-1}R(\frac{1}{x})=x^{p-1}+2x^{p-2}+\dots+(p-1)x+p$ implies the distinctness of the roots of $R(x)$. The distinctness of the roots of $S(x)$ comes from the fact that $S(x)$ and $S'(x)$ are co-prime. With $a=\frac{-p+1}{p(p+1)}x+\frac{1}{p}$ and $b=\frac{1}{p(p+1)}x^2-\frac{1}{p(p+1)}x$, we get $aS(x)+bS'(x)=1$. Therefore from Bezout $GCD(R(x),R'(x))=1$, which implies the result.
  
  For the second part of the lemma, a classical result due to Enestr\"om and Kakeya (see for example \cite{enestromkakeya}) concerning the bounds of the moduli of the zeros of polynomials having positive real coefficients states that all the complex roots of $R(x)$ have a moduli smaller or equal to $\frac{p-1}{p}$.
\end{proof}

\setcounter{tmp}{\value{theoremc}}
\setcounter{theoremc}{\value{lemmasmall}}
\begin{lemma}
\setcounter{theoremc}{\value{tmp}}
The value of $\overline{y}_n - \underline{y}_n$ decreases linearly: if $\underline{y}_n \neq \overline{y}_n$, then there is an integer $c$, with $0 \leq c \leq p$ such that $\overline{y}_{n+c} - \underline{y}_{n+c} < \overline{y}_n - \underline{y}_n$.
\end{lemma}
\begin{proof}
  Let $Y_n=(Y_{n_0},\dots,Y_{n_{p-1}})$. If $\underline{y}_n \neq \overline{y}_n$, that is, if the vector is not constant, the mean value is strictly between the greatest and smallest values: $\underline{y}_n < y_n < \overline{y}_n$. Consequently $\underline{y}_n < (Y_{n+1_{p-1}} = y_n + \frac{b_n}{p} \leq \overline{y}_n$. This reasoning applies while $\underline{y}_{n+i} \neq \overline{y}_{n+i}$. If we consider $p$ iterations, we either have $\underline{y}_{n+c}=\underline{y}_{n+c}$ for some $c \leq p$; or, we use the facts that $\underline{y}_n \leq \underline{y}_{n+1}$ and $\overline{y}_n \geq \overline{y}_{n+1}$, therefore we have $\underline{y}_n < Y_{n+p_i} \leq \overline{y}_n$ for all $i$ (because shifting the values leads to $Y_{n+p_i} = Y_{n+1+i_{p-1}}$ for all $0 \leq i \leq p-1$, and $\underline{y}_n \leq \underline{y}_{n+1+i} \leq Y_{n+1+i_{p-1}} \leq \overline{y}_{n+1+i} \leq \overline{y}_n$), thus $\underline{y}_n < \underline{y}_{n+p}$ and $\overline{y}_{n+p} \leq \overline{y}_n$.
\end{proof}

\begin{lemma}\label{lemma:eigenDM}
  The spectral radius of DM is strictly smaller than 1.
\end{lemma}
\begin{proof}
$M$ is a {\em compagnon matrix}, its characteristic polynomial is
$$x^p-\sum \limits_{k=0}^{p-1} \frac{1}{p} x^k = (x-1) R(x)$$
with $R(x)=x^{p-1} + \frac{p-1}{p}x^{p-2} + \dots + \frac{2}{p} x+ \frac{1}{p}$. From Lemma \ref{lemma:roots} we know that $R(x)$ has $p-1$ distinct roots $\lambda_1,\dots,\lambda_{p-1}$, all comprised between $\frac{1}{p}$ and $\frac{p-1}{p}$. The set of eigenvalues of $M$ is thus $M_\lambda=\{1,\lambda_1,\dots,\lambda_{p-1}\}$. We will prove that the set of eigenvalues of the matrix $DM$ is $DM_\lambda = \{0,\lambda_1,\dots,\lambda_{p-1}\}$.

Let $v_0,\dots,v_{p-1}$ be non null eigenvectors respectively associated to the eigenvalues $1,\lambda_1,\dots,\lambda_{p-1}$ of $M$. The case $v_1$ is particular and allows to conclude that $0$ is an eigenvalue of $DM$. The other eigenvectors of $M$ lead to the conclusion that $DM$ also admits the eigenvalues $\lambda_1,\dots,\lambda_{p-1}$.

\begin{itemize}
  \item $DM \, v_0 = D \,v_0$ since the associated eigenvalue is $1$, and $D \, v_0 = \0$ because the eigenspace associated to the eigenvalue 1 is the hyper plan of constant vectors
  . As a consequence, 0 is an eigenvalue of $DM$.
  \item For the other eigenvectors, that is, for $1 \leq i \leq p-1$, let $c_i$ be the constant vector with all its components equal to $\frac{1}{p}\sum_{k=0}^{p-1} v_{i_k}$, with $v_{i_k}$ the $k^{th}$ component of the vector $v_i$. $v_i - c_i \neq \0$, and
  $$\begin{array}{rcl}
  DM \, (v_i - c_i) & = & D \, (M \, v_i - M \, c_i)\\
  & = & D \, (\lambda_i \, v_i - c_i)\\
  & = & \lambda_i \, D \, v_i - D \, c_i\\
  & = & \lambda_i (v_i - c_i) - \0
  \end{array}$$
  where the last equality is obtained from the fact that by definition of $D$ we have $D \, v_i = v_i - c_i$. As a consequence, $\lambda_i$ is an eigenvalue of $DM$.
  \end{itemize}
  
  Finally, $DM_\lambda = \{ 0,\lambda_1,\dots,\lambda_p-1 \}$ and the spectral radius of $DM$ is smaller or equal to $\frac{p-1}{p}$.
\end{proof}


\subsection{\dots and allows to conclude}\label{ss:wave}

Lemma \ref{lemma:average} shows that there exists an $n$ in $\O(\log N)$ such that $Y_n$ is a constant vector. In this subsection, we prove that if $Y_n$ is a constant vector, then the shape of the sand pile configuration is wavy from the index $n$, exactly as it is stated in the main Theorem and on figure \ref{fig:wave1}.

\begin{lemma}\label{lemma:yn2wave}
  $Y_n$ is a constant vector of $\Z^p$ implies
  $$\pi(N)_{[n,\infty[} \in \left(\left(\bigplus \limits_{i=0}^{p+1} 0^i\right) \cdot p \cdot p\!-\!1 \cdot \ldots \cdot 1\right)^* 0^\omega$$
\end{lemma}

\begin{proof}
  From this index $n$ for which $Y_n$ is a constant vector, we will see that the sequence $(b_i)_{i \geq n}$ is determined, or more accurately nearly determined. The proof works as follows. We will first see that if $Y_n$ is a constant vector, then the value of $b_n$ is 0 or $p$. if it is $0$ then $Y_{n+1}$ is again a constant vector; if it is $p$, then the sequence $(b_i)_{n \leq i < n+p}$ is determined to be equal to $(p,p-1,\dots,1)$. Lemma \ref{lemma:plateau} states that there is no sequence of more than $p+1$ values 0 in the sequence $(b_i)_{i \in \N}$, hence we may have a sequence of at most $p+1$ values 0 (during which $Y_i$ remains constant), eventually followed by a value $p$ which triggers the pattern $(p,p-1,\dots,1)$. Finally, at the end of the sequence $(p,p-1,\dots,1)$, we are back to a constant vector $Y_i$, and the reasoning can be repeated until the end of the configuration (the $0^\omega$), hence the result.
   
  We concentrate on the sequence of values of $Y_i$. The fact that its components are integers, and especially the last one, will play a crucial role in the determination of the value of $b_i$ because $0 \leq b_i \leq p$ (let us recall that the sequence $b_i$ is the height difference representation of the fixed point with $N$ grains, {\em i.e.}, $b_i=\pi(N)_i$).

  We start from the hypothesis that $Y_n=(\alpha,\dots,\alpha)^T$, thus from the averaging system's equation (\ref{eq:averaging}) we have $Y_{n+1}=(\alpha,\dots,\alpha,\alpha+\frac{b_n}{p})^T$. $Y_{n+1}$ is an integer vector and $\alpha$ is an integer, hence $b_n$ equals 0 or $p$.
  \begin{itemize}
    \item If $b_n=0$ then $Y_{n+1}=(\alpha,\dots,\alpha)^T$ and we are back to the same situation, the dilemma goes on: the value of $b_{n+1}$ is not determined, it can be $0$ or $p$.
    \item If $b_n=p$ then $Y_{n+k+1}=(\alpha,\dots,\alpha,\alpha+1)^T$ from the relation above. A regular pattern then emerges:
    \begin{itemize} 
      \item if $Y_{n+1}=(\alpha,\dots,\alpha,\alpha+1)^T$, then $Y_{n+2}=(\alpha,\dots,\alpha,\alpha+1,\frac{p\alpha+1+b_{n+1}}{p})^T$ and it determines $b_{n+1}=p-1$ so that $Y_{n+2}=(\alpha,\dots,\alpha,\alpha+1,\alpha+1)$ is an integer vector;
      \item if $Y_{n+2}=(\alpha,\dots,\alpha,\alpha+1,\alpha+1)^T$, then $Y_{n+2}=(\alpha,\dots,\alpha,\alpha+1,\frac{p\alpha+2+b_{n+2}}{p})^T$ and it determines $b_{n+2}=p-2$ so that $Y_{n+3}=(\alpha,\dots,\alpha,\alpha+1,\alpha+1,\alpha+1)^T$ is an integer vector;
      \item {\em et cetera} we have $b_{n+i}=p-i$ for $0 \leq i < p$, and eventually $Y_{n+p}=(\alpha+1,\dots,\alpha+1)^T$ is a constant vector (note that $Y_0$ has a negative mean, hence $\alpha$ is negative, which is consistent with the $\alpha+1$ we obtain).
    \end{itemize}
  \end{itemize}
  
  Let us illustrate the developments above on the following picture, where arrows are labeled by values of $b_i$.
  \begin{center}
    \includegraphics{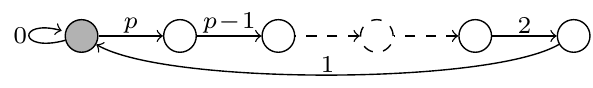}
  \end{center}
   When $Y_n$ is constant we are in the grey node, then it is either possible that $b_n=0$ in which case we are back in a situation where $Y_{n+1}$ is constant, or $b_n=p$ in which case $b_{n+1}=p-1,\dots,b_{n+p-1}=1$ and we are back in a situation where $Y_{n+p}$ is a constant vector. The last argument is that from Lemma \ref{lemma:plateau}, it is not possible to have a sequence of more than $p+1$ consecutive 0, hence it is not possible to loop more than $p+1$ times on the grey node.
   
   To conclude, if $Y_n$ is a constant vector we are on the grey node, and then the set of possibilities for the sequence $(b_i)_{n \leq i}$ is described by paths on the automata above. Furthermore it is not possible to loop more than $p+1$ consecutive times on the grey node, consequently the set of possibilities for the sequence $(b_i)_{n \leq i}$ is exactly the one described in the statement of the Theorem. Note that it is eventually possible to loop infinitely on the label 0, when there is no more grains on columns since Lemma \ref{lemma:plateau} won't apply anymore, hence the $0^\omega$. 
\end{proof}

\subsection{Proofs of Theorem \ref{theorem:main} and Theorem \ref{theorem:main2}}\label{ss:proof}

It remains to put the two pieces together to prove Theorem \ref{theorem:main}.

\setcounter{tmp}{\value{theoremc}}
\setcounter{theoremc}{\value{mytheorem}}
\begin{theorem}
\setcounter{theoremc}{\value{tmp}}
  There exists an $n$ in $\O(\log N)$ such that
  $$\pi(N)_{[n,\infty[} \in \left(\left(\bigplus \limits_{i=0}^{p+1} 0^i\right) \cdot p \cdot \ldots \cdot 2 \cdot 1\right)^*0^\omega$$
\end{theorem}

\begin{proof}
  From Proposition \ref{lemma:average}, there exists an $n$ in $\O(\log N)$ such that $Y_n$ is a constant vector. We can therefore apply Lemma \ref{lemma:yn2wave} for this $n$ and the result follows.
\end{proof}

An example of the various representations of fixed points ($X_n$, $Y_n$ and wavy shape) for $\pi(2000)$ with $p=4$ is given in appendix \ref{a:2000}.

In order to prove the refinement of Theorem \ref{theorem:main2}, it seems necessary to overcome the ``static'' study --- direct, for any fixed point --- presented here and consider the dynamic of sand grains on columns beyond $\O(\log N)$, from $\pi(0)$ to $\pi(N)$. We recall that $\pi(N)$ can be computed inductively, using the relation
$$\text{for all } k>0,~ \pi(\pi(k-1)^{\downarrow 0})=\pi(k)$$
where $\sigma^{\downarrow 0}$ denotes the configuration obtained by adding one grain on column 0 of $\sigma$. We start from $\pi(0)$ and inductively compute $\pi(1),\pi(2),\dots,\pi(N\!-\!1)$ and $\pi(N)$ by repeating the addition of one grain on column 0. The sequence of firings from $\pi(k\!-\!1)^{\downarrow 0}$ to $\pi(k)$ is called the $k^{th}$ {\em avalanche} (see subsection \ref{ss:hourglass}). 

An interesting Corollary of Theorem \ref{theorem:main}, proven below, is that for any $p$ and $N$, there is an index $n$ in $\O(\log N)$ such that the $N^{th}$ avalanche fires a set of consecutive columns (without missing any) on the right of that index $n$ (conjectured in \cite{LATA} and proven only for $p=2$). We now formally introduce this important property on avalanches, which, when it is verified starting from an index $n$, leads to regularities in the avalanche process beyond column $n$ (see \cite{LATA} and \cite{MFCS} for details). We will see that a Corollary of Theorem \ref{theorem:main} is that this property is verified on the $k^{th}$ avalanche starting from an index in $\O(\log k)$.

We say that there is a {\em hole} at position $i$ in an avalanche $s^k$ if and only if $i \notin s^k$ and $(i\!+\!1) \in s^k$. An interesting property of an avalanche is the absence of hole from an index $l$, which tells that,
$$\text{there exists an } m \text{ such that }\begin{array}[t]\{{l}. \text{for all } i \text{ with } l \leq i \leq m, \text{ we have } i \in s^k\\ \text{for all } i \text{ with } m < i, \text{ we have } i \notin s^k\end{array}$$
namely, from column $l$, a set of consecutive columns is fired, and nothing else. We say that an avalanche $s^k$ is {\em dense starting from} an index $l$ when $s^k$ contains no hole $i$ with $i \geq l$. We have already explained in \cite{LATA} that this property induces a kind of ``pseudo linearity'' on avalanches, that it somehow ``breaks'' the criticality of avalanche's behavior and let them flow smoothly along the sand pile. Let us introduce a formal definition:

\begin{definition}
  $\mathcal L'(p,k)$ is the minimal column such that the $k^{th}$ avalanche is dense starting at $\mathcal L'(p,k)$:
  $$\mathcal L'(p,k)=\min \{ l \in \N ~|~ \exists m \in \N \text{ such that } \forall l \leq i \leq m, i \in s^k \text{ and } \forall i > m, i \notin s^k \}$$
  Then, the global density column $\mathcal L(p,N)$ is defined as:
  $$\mathcal L(p,N)=\max \{\mathcal L'(p,k) ~|~ k \leq N \}$$ 
\end{definition}
The global density column $\mathcal L(p,N)$ is the smallest column number starting from which the $N$ first avalanches are dense (contain no hole). A Corollary of Theorem \ref{theorem:main} is:

\begin{corollary}\label{corollary:snowball}
For all parameter $p$, $\mathcal L(p,N)$ is in $\O(\log N)$.
\end{corollary}
\begin{proof}
  Let $p$ be a fixed parameter. In order to prove the result, we will prove that if
  $$\pi(k)_{[n,\infty[} \in \left(\left(\bigplus \limits_{i=0}^{p+1} 0^i\right) \cdot p \cdot \ldots \cdot 2 \cdot 1\right)^*0^\omega$$
  then the next avalanche $s^{k+1}$ is dense starting from $n\!+\!c$ with $c \leq p\!+\!1$, {\em i.e.}, $\mathcal L'(p,k+1) \leq n\!+\!c$. Since such an $n$ is in $\O(\log k)$ for all $k$ from Theorem \ref{theorem:main}, the Corollary follows.
  
  Let us consider the following case disjunction.
  \begin{itemize}
    \item If the $k\!+\!1^{th}$ avalanche ends before column $n$, formally if $\max s^{k+1} < n$, then obviously $\mathcal L'(p,k+1) < n$ because for all $l$, $\mathcal L'(p,l) \leq \max s^l$.
    \item If the $k\!+\!1^{th}$ avalanche ends beyond column $n$, formally if $\max s^{k+1} \geq n$, then we will consider the dynamic of the avalanche process $s^{k+1}$ on columns greater than $n$. There is an intuitive notion of {\em time} during an avalanche, induced by the order of the firings: $s^{k+1}_t$ is the fired column at time $t$. For the sake of clarity, let us precise our hypothesis given by Theorem \ref{theorem:main} on the shape of $\pi(k)_{[n,\infty[}$. Let $l$ and $z_1,\dots,z_l$ be the integers such that
    $$\pi(k)_{[n,\infty[} = \overset{z_1}{\overbrace{0 \cdot \ldots \cdot 0}} \cdot p \cdot \ldots \cdot 2 \cdot 1 \cdot \overset{z_2}{\overbrace{0 \cdot \ldots \cdot 0}} \cdot p \cdot \ldots \cdot 2 \cdot 1 \cdot \dots \cdot \overset{z_l}{\overbrace{0 \cdot \ldots \cdot 0}} \cdot p \cdot \ldots \cdot 2 \cdot 1\cdot 0^\omega$$
    
    The aim is to prove that $s^{k+1}$ is dense starting from $n+z_1$, which is proven by induction on the following Statement:
    \begin{statement}
      For all $i \geq n+z_1$, if $i \notin s^{k+1}$ then $(i\!+\!1) \notin s^{k+1}$.
    \end{statement}
    Let us prove this Statement for all $i \geq n+z_1$, by considering two cases.
    \begin{itemize}
      \item If $\pi(k)_i > 0$, we prove the contraposition of the Statement: if $(i\!+\!1) \in s^{k+1}$ then $i \in s^{k+1}$. The firing at $(i\!+\!1)$ gives $p$ units of height difference to column $i$, which becomes unstable and thus $i \in s^{k+1}$.
      \item If $\pi(k)_i = 0$, we prove the Statement by contradiction: let us suppose that $(i\!+\!1) \in s^{k+1}$. Column $i$ can: receive $p$ units of height difference from $(i\!+\!1)$, and 1 unit of height difference from $(i\!-\!p)$. By hypothesis, $i$ receives $p$ units of height difference thus $(i\!-\!p) \notin s^{k+1}$ otherwise $i$ receives $p\!+\!1$ units of height difference during the avalanche process and $i \in s^{k+1}$.
      
      According to the shape of $\pi(k)$ beyond column $n$, there are two possibilities for the value of $\pi(k)_{i+1}$. Let us show that in any case we have $(i\!-\!p\!+\!1) \in s^{k+1}$.
      \begin{itemize}
        \item If $\pi(k)_{i+1}=0$, then since $(i\!+\!1)$ is fired it should receive at least $p\!+\!1$ units of height difference, which is possible only if $(i\!-\!p\!+\!1) \in s^{k+1}$.
        \item If $\pi(k)_{i+1}=p$, then we prove it in two steps. Firstly, from the locality of the rule (it involves columns at distance at most $p$) and the stability of $\pi(k)$, a column can't be fired if none of the $p$ preceding columns has been fired. Secondly, for all $j$ with $0 < j < p$, we have $\pi(k)_{i+1-j} < p$ thus by induction on fact that such a column can't be fired if its successor (the consecutive column on its right) has been fired, $i\!+\!1$ is fired before any of the $p$ preceding columns. Consequently, it is not possible to fire any column greater than $i\!+\!1$ before column $i\!+\!1$ itself. We can therefore conclude that $i\!+\!1$ received units of height difference from the left of the configuration, that is, from $(i\!-\!p\!+\!1)$.
      \end{itemize}
      We proved that $(i\!-\!p\!+\!1) \in s^{k+1}$. Now, if $(i\!+\!1) \in s^{k+1}$ then the $z_\iota$ associated to column $i$ is at most $p\!-\!1$, otherwise the avalanche process can't go beyond this flat surface (again, from the locality of the rule). As a consequence, $\pi(k)_{i-p} > 0$, and since $(i\!-\!p\!+\!1) \in s^{k+1}$ it gives $p$ units of height difference to $i\!-\!p$ thus it becomes unstable, but $(i\!-\!p) \notin s^{k+1}$, which is a contradiction.
    \end{itemize}
  \end{itemize}  
\end{proof}

\setcounter{tmp}{\value{theoremc}}
\setcounter{theoremc}{\value{mytheorem2}}
\begin{theorem}
\setcounter{theoremc}{\value{tmp}}
  There exists an $n$ in $\O(\log N)$ such that
  $$\pi(N)_{[n,\infty[} \in (p \cdot \ldots \cdot 2 \cdot 1)^*[0](p \cdot \ldots \cdot 2 \cdot 1)^*0^\omega$$
  where $[0]$ stands for at most one symbol 0.
\end{theorem}

\begin{proof}
  We prove a slightly stronger result, that there is an $n$ in $\O(\log N)$ such that
  $$\pi(N)_{[n,\infty[} \in \left(\bigplus \limits^{p+1}_{i=0} 0^i\right)(p\cdot \ldots \cdot 2 \cdot 1)^*[0](p\cdot \ldots \cdot 2 \cdot 1)^*0^\omega$$
  
  The proof works by induction on $N$, with the base case coming from Theorem \ref{theorem:main}. For the induction, suppose the result holds for $\pi(N)$ with the wavy pattern $0^i(p \cdot \ldots \cdot 1)^x0(p \cdot \ldots \cdot 1)^y$, with $i \geq 0$, $x>0$ and $y \geq 0$, beginning at $n$ in $\O(\log N)$. In addition, there is a minimal index $l$ such that the $(N\!+\!1)^{th}$ avalanche fires a set of consecutive columns, from $l$ to $m$, and from Corollary \ref{corollary:snowball} this $l$ is in $\O(\log N)$. Without loss of generality, we consider that $l \leq n-p$.
  \begin{itemize}
    \item if $m < n+i$, either the wavy pattern is not affected (if $m \leq n-p$), or Theorem \ref{theorem:main} ensures that there will always exist an index $n$ in $\O(\log N)$ such that the wavy pattern is preserved;
    \item if $n+i \leq m$, we first note that $i \neq p+1$ because an avalanche cannot ``cross'' a set of $p+1$ consecutive 0. In this case, the avalanche triggers a chain reaction that stops on the last wave before the 0 is encountered, and thus spans the $x$ first waves, {\em i.e.}, $m=n+(x-1)p$. A simple look at the transition rule indicates that the firing of columns $n-p$ to $m$ shifts the wavy pattern to the following: $0^i(p \cdot \ldots \cdot 1)^{x-1}0(p \cdot \ldots \cdot 1)^{y+1}$.
  \end{itemize}
  In any case, the result holds for $\pi(N+1)$. 
\end{proof}



\section{Conclusive discussion}\label{s:conclusion}

The result of this paper stresses the fact that sand pile models are at the edge between discrete and continuous systems. Indeed, when there are very few sand grains, each one seems to contribute greatly to the global shape of the configuration. However, when the number of grains is very large, a particular sand grain seems to have no importance to describe the shape of a configuration. The result can be interpreted has a separation of the discrete and continuous parts of the system, despite it is not clearly stated. On one hand, the left seemingly unordered part, interpreted as reflecting the discrete behavior, prevents regularities to emerge. On the other hand, the right and ordered part, interpreted as reflecting the continuous behavior, lets a regular and smooth pattern to come into view.

Nevertheless, the separation between discrete and continuous behaviors may be challenged because the continuous part emerges from the discrete part. We have two remarks about this latter fact. Firstly, the consequence seems to be a slight bias appearing on the continuous part: it is not fully homogeneous ---that is, with the exact same height difference at each index--- which would have been expected for a continuous system, but a ---very small--- pattern is repeated. This remark actually seems to be misleading because this exact same bias disappears when we consider an initial configuration with $p$ consecutive indices with $N$ units of height difference, thus it looks like the bias comes from the gap between the unicity of the initial column compared to the parameter $p$ of the rule. Secondly, if we consider the asymptotic form of a fixed point, the relative size of the discrete part is null. This, regarding the intuition described above that when the number of grains is very large then a particular grain has no importance, is satisfying.

Let us save the last words to a distracting application to a famous paradox. Someone who has a very little amount of money is called {\em poor}. Someone {\em poor} who receives one cent remains {\em poor}. Nonetheless, if the increase by 1 cent is repeated a great number of time then the person becomes {\em rich}. The question is: when exactly does the person becomes {\em rich}? An answer may be that {\em richness} appears when money creates waves...

\bibliographystyle{alpha}
\bibliography{biblio}

\newpage


\appendix

%
%

\section{There is no plateau of length larger than $p+1$}\label{a:plateau}

A plateau is a set of at least two non empty and consecutive columns of equal height. The length of a plateau is the number of columns composing it.

\begin{lemma}\label{lemma:plateau}
  For any $N$ and any configuration $\sigma$ such that $(N,0^\omega) \overset{*}{\to} \sigma$, in $\sigma$ there is no plateau of length strictly greater than $p+1$.
\end{lemma}

\begin{proof}
  This proof proceeds by contradiction, using the fact that configurations are sequences of non-negative integers ($\mathcal H_1$). Suppose there exists a configuration $\sigma$ reachable from $(N,0^\omega)$ for some $N$, such that there is a plateau of length at least $p+2$ in $\sigma$. Since there is no plateau in the initial configuration $(N,0^\omega)$, and there is a finite number of steps to reach $\sigma$, there exists two configurations $\rho$ and $\tau$ such that $\rho \to \tau$ and such that there is a plateau of length at least $p+2$ in $\tau$, and none in $\rho$ ($\mathcal H_2$).
  Let $k$ be the leftmost column of the plateau of length at least $p+2$ in $\tau$, {\em i.e.} for all column $j$ between $k$ and $k+p$, $\tau_j=0$ ($\mathcal H_3$). We will now see that there is no $i$ such that $\rho \overset{i}{\to} \tau$, which completes the proof.
  \begin{itemize}
    \item if $i<k-p$ or $i>k+p+1$ then a firing at $i$ has no influence on columns between $k$ and $k+p+1$ and there is a plateau of length at least $p+2$ in $\rho$, contradicting $\mathcal H_2$.
    \item if $k-p \leq i \leq k$ then according to the rule definition we have $\tau_{i+p}=\rho_{i+p-1}-1$, and from $\mathcal H_3$ $\rho_{i+p}=0$ therefore $\tau_{i+p}<0$ which is not possible (recalled in $\mathcal H_1$).
    \item if $k < i \leq k+p+1$ then according to the rule definition we have $\tau_{i-1}=\rho_{i-1}-p$ and from $\mathcal H_3$ $\rho_{i-1}=0$ therefore $\tau_{i-1}<0$ which again is not possible from $\mathcal H_1$.
  \end{itemize}
\end{proof}

%
%

\section{The support of $\pi(N)$ is in $\Theta(\sqrt{N})$}\label{a:support}

We give bounds for the maximal index of a non-empty column in the fixed point $\pi(N)$ according to the number $N$ of grains, denoted $w(N)$. $w(N)$ can be interpreted as the support or width or size of $\pi(N)$. We consider a general model KSPM($p$) with $p$ a constant integer greater or equal to 1. A formal definition of $w(N)$ is for example $w(N)=w(\pi(N))=\min\{ i | \forall j \geq i, \pi(N)_j=0 \}$. See figure \ref{fig:frame}.

\begin{lemma}\label{lemma:support}
  The support of $\pi(N)$ is in $\Theta(\sqrt{N})$.
\end{lemma}

\begin{proof}
  The support of $\pi(N)$ is denoted $w(N)$.
  
  Lower bound: $\pi(N)$ is a fixed point, therefore by definition for all index $i$ we have $\pi(N)_i \leq p$. Then,
  $$N \leq \sum \limits_{i=0}^{w(N)} p \cdot i = p \frac{w(N) \cdot (w(N)+1)}{2} < p^2(w(N)+1)^2$$
  hence $\frac{1}{p} \sqrt{N} - 1 < w(N)$.
  
  Upper bound: From Lemma \ref{lemma:plateau}, there is no plateau of length greater than $p+1$. Therefore, for $w(N) \geq p$ we have
  $$N \geq \sum \limits_{i=0}^{\lfloor \frac{w(N)}{p+1} \rfloor} (p+1) \cdot i \geq (p+1)\left(\frac{\left(\frac{w(N)}{p+1} - 1\right)\frac{w(N)}{p+1}}{2}\right) > \left(\frac{w(N)}{p+1} - 1\right)^2$$
  hence $(p+1)\sqrt{N} + p+1 > w(N)$.
\end{proof}

\begin{figure}
  \centering \includegraphics{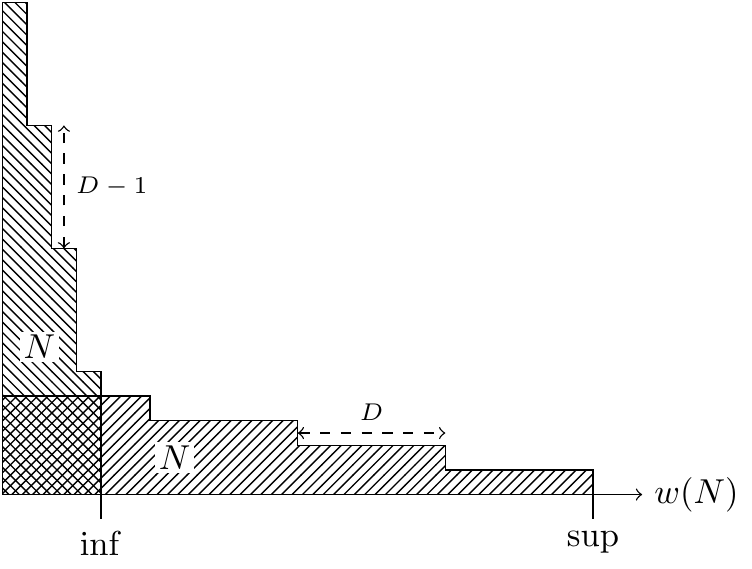}
  \caption{The support of $\pi(N)$ is in $\Theta(\sqrt{N})$. It is lower bounded by the fact that each height difference on a stable configuration is at most $p$, and upper bounded by the fact that there is no plateau of length larger than $p+1$.}
  \label{fig:frame}
\end{figure}
   
%
%

\section{$\pi(2000)$ for $p=4$}\label{a:2000}

Figure \ref{fig:2000} presents some representations of $\pi(2000)$ for $p=4$ used in the developments of the paper: differences of shot vector (figure \ref{fig:2000-dsv}), shot vector (figure \ref{fig:2000-sv}) and height (figure \ref{fig:2000-h}).
$$\pi(2000)=\begin{array}[t]{l}(4,0,4,1,3,2,4,1,1,3,4,3,4,2,0,1,4,2,2,1,\\
~~~~~~~~~4,3,2,1,0,4,3,2,1,4,3,2,1,4,3,2,1,4,3,2,1,0^\omega)\end{array}$$
We can notice on figure \ref{fig:2000-dsv} that the shot vector differences contract towards some ``steps'' of length $p$, which corresponds to the statement of Lemma \ref{lemma:average} that the vector $Y_n$ becomes constant exponentially quickly (note that this graphic plots the opposite of the values of the components of $Y_n$). The shot vector representation on figure \ref{fig:2000-sv} corresponds to the values of the components of $X_n$, which we did not manage to tackle with classical methods. Figure \ref{fig:2000-h} pictures the sand pile configuration on which the wavy shape appears starting from column 20.

\begin{landscape}

\begin{figure}
  \centering
  \subfloat[$\pi(2000)$ for $p=4$ represented as a sequence of shot vector differences.]{\label{fig:2000-dsv}\includegraphics[scale=0.75]{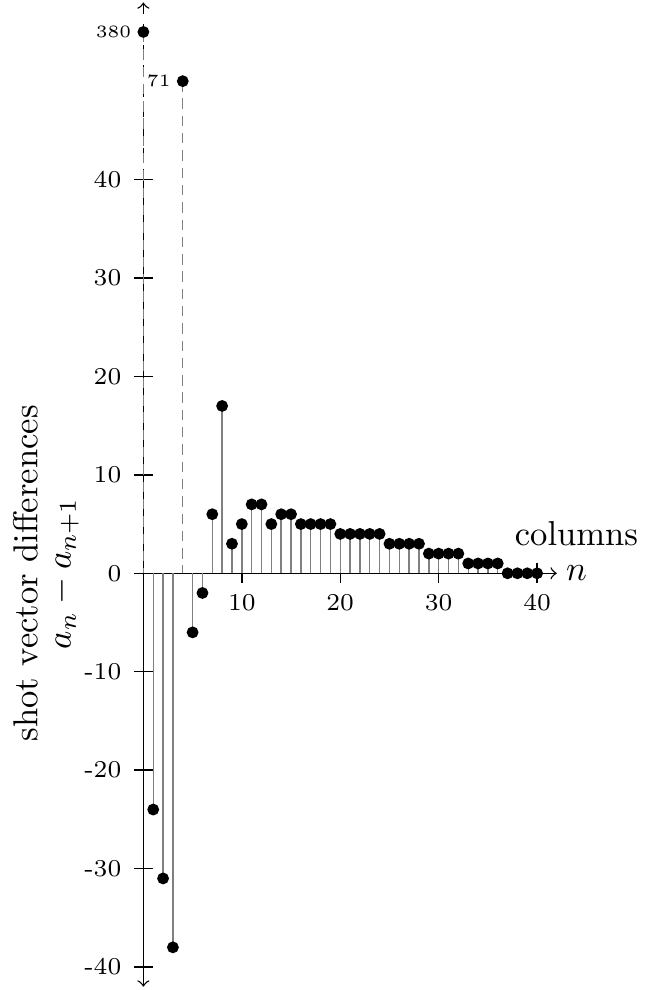}}
  \hspace{.5cm}              
  \subfloat[$\pi(2000)$ for $p=4$ represented by its shot vector.]{\label{fig:2000-sv}\includegraphics[scale=0.75]{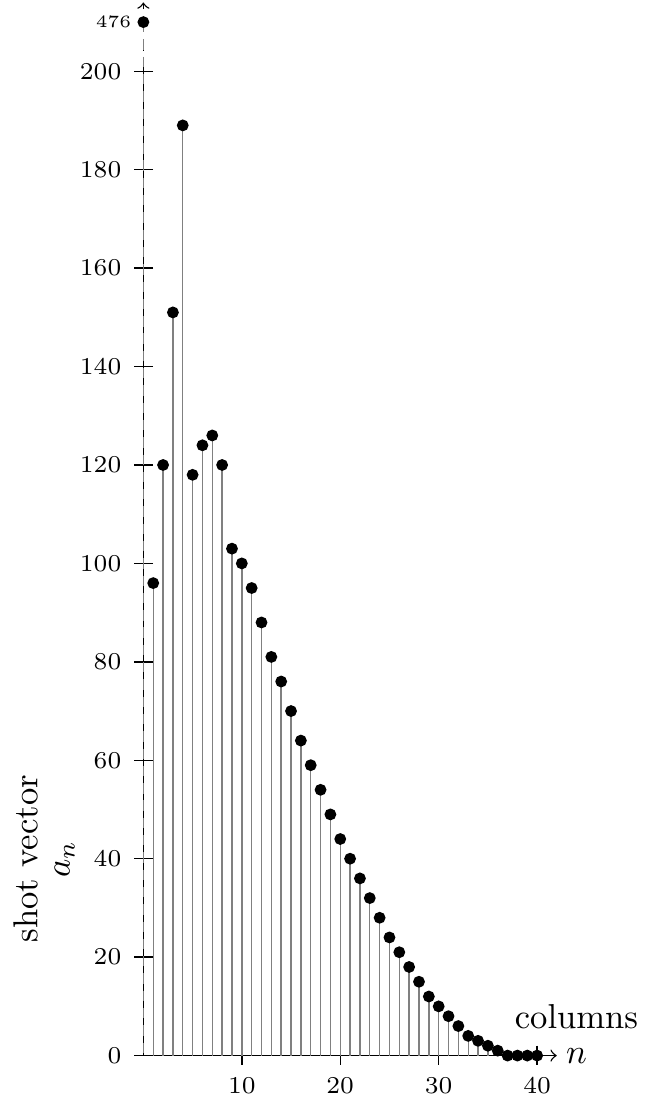}}
  \hspace{.5cm}
  \subfloat[$\pi(2000)$ for $p=4$ represented as stacked sand grains.]{\label{fig:2000-h}\includegraphics[scale=0.75]{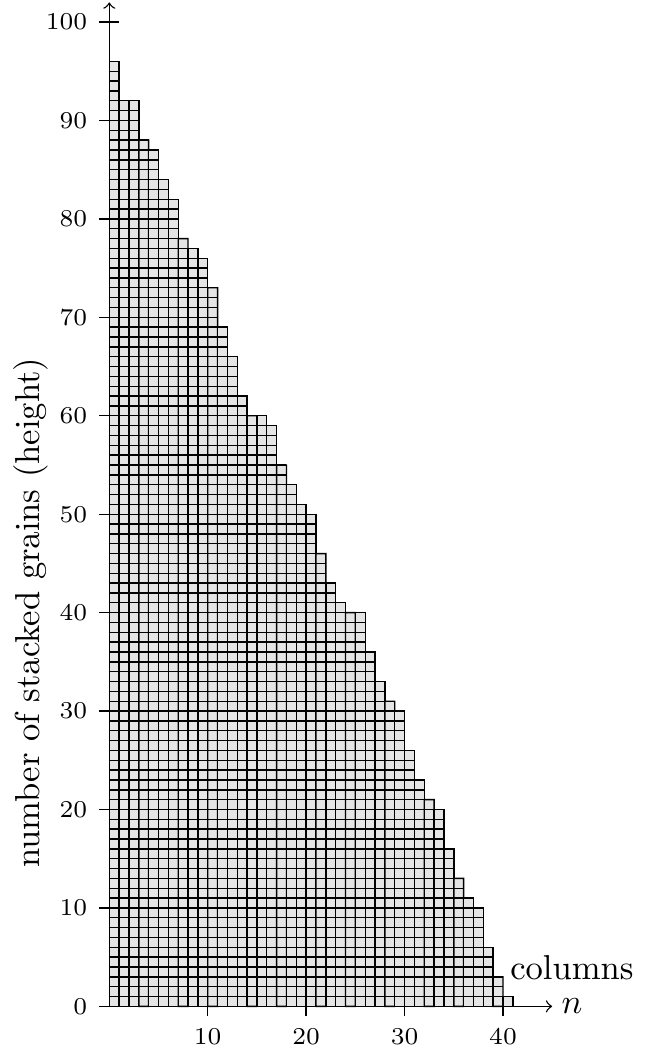}}
  \caption{Representations of $\pi(2000)$ for $p=4$.}
  \label{fig:2000}
\end{figure}

\end{landscape}

\end{document}